\documentclass[journal,10pt]{IEEEtran}

\usepackage{color}  
\usepackage[final]{graphicx} 
\usepackage{amsmath,amssymb,amscd,latexsym,dsfont,amsthm,algorithm,algorithmic,amsbsy,epsfig,bbm,mathrsfs,fancyhdr,fancyvrb} 

\usepackage{upgreek,textgreek}
\usepackage{auto-pst-pdf}
\usepackage{caption}
\usepackage{psfrag}
\usepackage{epstopdf} 
\usepackage{tabls}
\usepackage[noadjust]{cite}
\usepackage{color,comment}
\usepackage{balance}
\usepackage{float}
\usepackage{subfigure}
\usepackage{placeins}
\usepackage[top=0.75in, bottom=1in, left=0.625in, right=0.625in]{geometry}

\newcommand{\tr}[1]{\mathrm{#1}}

\newcounter{tempEquationCounter}
\newcounter{thisEquationNumber}
{\setcounter{equation}{\value{tempEquationCounter}}
\end{figure*}
}%

\newtheorem{theorem}{Theorem}
\newtheorem{lemma}{Lemma}
\newtheorem{proposition}{Proposition}

\begin{document}
\title{Limits on the Capacity of In-Band Full Duplex Communication in Uplink Cellular Networks}

\author{%
Itsikiantsoa Randrianantenaina, Hesham Elsawy, and~Mohamed-Slim~Alouini.\\
\thanks{%
I. Randrianantenaina, H. Elsawy and M.-S. Alouini are with the Computer, Electrical, and Mathematical Science and Engineering (CEMSE) Division, King Abdullah University of Science and Technology (KAUST), Thuwal, Saudi Arabia. [e-mails: \{itsikiantsoa.randrianantenaina, hesham.elsawy, slim.alouini\}@kaust.edu.sa].}%
\thanks{%
This work was funded in part by King Abdullah University of Science and Technology (KAUST).
}%
}%

\maketitle
\thispagestyle{empty}
   
\begin{abstract}

Simultaneous co-channel transmission and reception, denoted as in-band full duplex (FD) communication, has been promoted as an attractive solution to improve the spectral efficiency of cellular networks. 
However, in addition to the self-interference problem, cross-mode interference (i.e., between uplink and downlink) imposes a major obstacle for the deployment of FD communication in cellular networks. More specifically, the downlink to uplink interference represents the performance bottleneck for FD operation due to the uplink limited transmission power and venerable operation when compared to the downlink counterpart. While the positive impact of FD communication to the downlink performance has been proved in the literature, its effect on the uplink transmission has been neglected. This paper focuses on the effect of downlink interference on the uplink transmission in FD cellular networks in order to see  whether FD communication is beneficial  for the uplink transmission or not, and if yes for which type of network. To quantify the expected performance gains, we derive a closed form expression of  the maximum  achievable uplink capacity in FD cellular networks. In contrast to the downlink capacity which always improves with FD communication, our results show that the uplink performance may improve or degrade depending on the associated network parameters. Particularly, we show that the intensity of base stations (BSs) has a more prominent effect on the uplink performance than their transmission power.






\end{abstract} 
   

\section{Introduction}

The fifth generation (5G) of cellular networks are expected to provide a tangible evolution in the network performance and supported services. In terms of network performance, it is expected to achieve a 1000-fold capacity improvement with at least 100-fold leap in the peak data rate when compared to the state-of-the-art 4G cellular systems.  Hence, cellular network operators are seeking all means to increase the spectrum utilization and cope with the ambitious 5G requirements. In this regards, in-band full-duplex (FD) communication has captivated the attention due to its potential to double the per-link spectral efficiency. FD communication allows transceivers to transmit and receive on the same time-frequency resource block, making traditional duplexing techniques (e.g., frequency division duplexing) obsolete and providing a better utilization for the available spectrum. In general, the implementation of FD communication relies on the self-interference (SI) cancellation capabilities that allow transceivers  to keep an acceptable isolation between the transmit and receive circuits while transmitting and receiving on the same time-frequency resource block, thanks to the recent advances in transceiver design~\cite{Hong2014}. Therefore, in order  harvest the expected spectral efficiently gain from  viable FD communication, tremendous efforts are invested to enhance the SI cancellation \cite{Lee2015, Sabharwal2014}.

In the context of cellular networks, instead of dividing the entire bandwidth (BW) (i.e., in time or in frequency) between uplink and downlink transmission, FD communication improves the spectrum utilization by allowing uplink and downlink to simultaneously occupy the entire bandwidth. This makes the implementation of FD communication in cellular networks disputable, even if perfect SI is achieved, due to the imposed cross-mode interference (i.e., interference between uplink and downlink transmissions).  The studies in  \cite{Alves2014, Lee2015 , Alammouri2015} show that, despite the imposed cross-mode interference, FD communication can effectively improve the spectral efficiency for the downlink. However, its effect on the uplink performance has been overlooked. We argue that the performance gain provided by FD communication in the downlink may come on the expense of a significant degradation in the uplink performance, due to the high disparity between the uplink and downlink operation. For instance, the downlink BS with high transmission power capability may create severe interference  on the uplink due to the limited transmit power of the users' equipment (UEs). To support our argument, we derive an upper-bound of the achieved uplink rate under FD operation. In addition, the obtained upper-bound is exploited to benchmark power control and interference cross-mode management techniques in FD cellular networks. The achieved upper-bound with the power control is compared to the traditional half-duplex (HD) communication, obtain in \cite{ElSawy2014}, to visualize the maximum gain that could be achieved via FD communication in cellular networks. The results confirms that FD can have significant negative impact on uplink transmissions, especially in a macro network environment where the base stations (BSs) are characterized with their high transmission powers. Nevertheless, FD communication can improve the uplink capacity in small cell networks which are characterized by their lower transmission power, in spite of their higher density, when compared to the macro cell network.

This work is focused on the effect of in-band FD communication on the uplink transmission due to two reasons. First, the explicit effect of FD communication on the uplink transmission has been neglected  in the literature. Second,  uplink is the bottleneck for FD communication due to the limited transmission power of the UEs and the strong downlink interference  \cite{Alammouri2015}. Therefore, depending on the network parameters, this paper shows  whether FD communication is beneficial for uplink transmission or not, and if yes, quantifies the gain in terms of data rate. Since recent studies show that the BSs in cellular network exhibit random patterns rather than idealized grids, the analysis in this paper is based on stochastic geometry toolset. More particularly, we assume a single-tier cellular network in which the BSs locations constitute a Poisson point process (PPP) with intensity $\lambda$. Beside simplifying the analysis, the PPP has shown to give accurate estimation for the actual cellular network performance \cite{tractable_app, martin_ppp}. For the sake of tractability, the analysis is based on the hybrid approach proposed in \cite{Heath2013} for cellular networks, which assume a circular BS coverage with area $\lambda^{-1}$ for the test BS and a PPP for the interfering BSs. The hybrid model has been also used in \cite{Jeff_D2D} and showed a close match to the PPP simulation. Further, the downlink interference is approximated with a Gamma distribution using second moment matching obtained via stochastic geometry analysis for the BSs interference. The accuracy Gamma approximation for the aggregate interference in the hybrid model has been validated in \cite{Heath2013}. We further validate the Gamma approximation via simulation. We derive an upper-bound on the uplink rate which is based on Shannon's formula and rate-optimal water filling power control assuming that the uplink  transmitter has perfect instantaneous information about the interference as well as channel state information (CSI) towards the receiving BS.  


 The rest of this paper is organized as follows, the system model and methodology of analysis are presented in Section~\ref{sec:SystemModel}. Then,  Section~\ref{sec:InterferenceApproximation} details  the second moment matching approximation of the interference with a Gamma distribution. In Section~\ref{sec:Optimization}, we formulate and solve the optimization problem to get a closed-form expression of the maximum achievable uplink capacity. Numerical and simulation results are showed in Section~\ref{sec:Results} before concluding the paper in Section~\ref{Conclusion}.

\section{System Model}\label{sec:SystemModel}
\subsection{Network Model}

We consider a single-tier cellular network in which the BSs' locations follow a PPP $\Psi = \{ x_i \}$, where $ x_i \in \mathbb{R}^2$ denotes the $i^{th}$ BS location, with intensity $\lambda$. All BSs are equipped with single antennas and transmit with a constant power $P_{BS}$.  The UEs follow a stationary point process $\Phi$, where each user has an average power constraint $\bar{P}$. UEs associate  to the BS that provides the highest average signal strength. In this case, it is easy to show that the average distance between a user and his serving BS is given by $\bar{R} = {1}/({2 \sqrt{\lambda}})$. We focus the analysis on the uplink connection of a test cellular user located at the average distance $\bar{R}$ who is communicating with the test BS located at the origin. According to Slivnyak's theorem there is no loss in generality to focus on the performance of the BS located at the origin.    Following \cite{Jeff_D2D, Heath2013},  we approximate the test BS coverage area by a circle with a radius of $r_0 = \frac{1}{\sqrt{\pi \lambda}}$, to satisfy the average area of $\lambda^{-1}$ per BS. 

We assume general power-law path loss model in which the transmitted signal power decays at the rate $r^{-\eta}$ with the distance $r$, where $\eta > 2$ is  the path loss exponent depending on the environment. In addition to the path loss attenuation, the signal powers experience i.i.d. Gamma distributed channel gains, denoted by $\alpha \sim Gamma(m, \Omega)$ where $m$ and $\Omega$ are, respectively, the shape parameter and mean, \textit{i.e}, the  probability density function {\em pdf} of $\alpha$ is given by
\begin{equation}
f_{\alpha}(\alpha)=\frac{1}{\Gamma(m)}\left(\frac{m}{\Omega}\right)^{m}\alpha^{m-1}\exp\left(-\frac{m\alpha}{\Omega}\right); \;\;\;\;\;\alpha > 0.
\end{equation}
We allow the intended link to have a different fading coefficient $\alpha_0 \sim Gamma(m_o, \Omega_o)$. We ignore the effect of the SI and the uplink network interference on the FD uplink capacity as we intend to focus on the downlink to uplink interference and obtain an upper bound on the achievable uplink capacity. In this case, the upper-bound on the uplink capacity, which is defined by Shannon's capacity, is given by:
\begin{equation}
C^*_{FD}= B \mathbb{E}\left[ \log_2\left(1+ \frac{P^* h}{ \underset{x_i \in \Psi}{\sum} P_{BS} \alpha_i \left\| x_i \right\|^{-\eta} +N_{0}}\right)\right]
\label{FD_capacity}
\end{equation}
where $B$ is the total  BW occupied by uplink transmission, $P^*$ is the optimal transmit power, $h = \frac{\alpha_o}{(2\sqrt{\lambda})^\eta} \sim Gamma(m_o, {\Omega_o}/{(2\sqrt{\lambda})^\eta})$ is the intended channel gain including fading and path-loss attenuation, $\left\| .\right\|$ denotes the Euclidean norm, and $N_0$ is the noise power.

\begin{figure}[t]
\psfrag{spl}[c][c][0.6]{~\! $\frac{1}{\sqrt{\pi\lambda}}$}
\psfrag{sl}[c][c][0.6]{~\! $\frac{1}{2\sqrt{\lambda}}$}
\psfrag{User served by the BS}[c][c][0.8]{User served by the BS}
\psfrag{Interfering link from other BSs}[c][c][0.8]{Interfering link from other BSs}
\psfrag{Interfering BSs}[c][c][0.8]{BSs}
\psfrag{Uplink link}[c][c][0.8]{Uplink link}
\begin{center}
\scalebox{0.7}{\includegraphics[width=1\columnwidth]{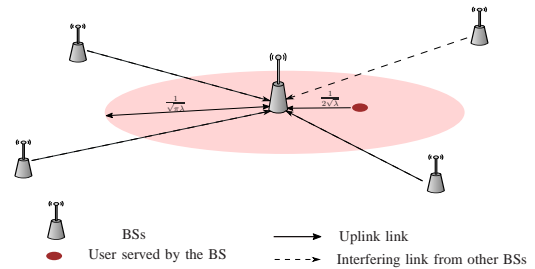}}
\end{center}
\caption{The hybrid system model with circular test BS coverage and PPP interfering BSs.}	
\label{fig:SystemModel}
\end{figure}

 \subsection{Methodology of Analysis}

In this paper, we aim to derive the optimal power control policy that maximizes the uplink capacity given in \eqref{FD_capacity}. Following \cite{sousa}, the optimal power policy and the maximum capacity can be obtained by solving the following optimization poblem 
	\begin{equation}
\begin{aligned}\label{eq:Optimization_problem}
& \underset{P(h,I)}{\tr{max}} C = B \mathbb{E} \left[ \log_2\left( 1+ \frac{P(h,I) h }{I+N_0} \right) \right] \\
& \tr{s.t.~~} \mathbb{E}[P(h,I)] \leq \overline{P}.
\end{aligned}
\end{equation}
where  $P(h,I)$, is the optimization variable denoting the test user's instantaneous transmit power, $I = \underset{x_i \in \Psi}{\sum} P_{BS} \alpha_i \left\| x_i \right\|^{-\eta}$ is the aggregate downlink to uplink interference, and the expectations in \eqref{eq:Optimization_problem} is computed w.r.t. both $h$ and $I$ as:
\begin{align}\label{eq:averageCapacity}
&C^{*}_{FD}=B\int_{h}\!\int_{I}\log_{2}\left(1+\frac{hP^*(h,I)}{I+N_{0}}\right) f_{h}(h)f_{I}(I)\mathrm{d}h\mathrm{d}I, \quad \text{and} \notag \\
&\mathbb{E}[P(h,I)] = \int_{h}\!\int_{I}P^*(h,I)f_{h}(h)f_{I}(I)\mathrm{d}h\mathrm{d}I
\end{align} 
where $f_I(.)$ and $f_h(.)$ are the {\em pdfs} of, respectively, $I$ and $h$.  Since there is no closed form expression for the {\em pdf} of $I$ \cite{survey}, we can neither compute the averages in \eqref{eq:averageCapacity} nor solve the optimization problem in \eqref{eq:Optimization_problem}. To overcome this problem, we approximate the downlink interference via a gamma distributed random variable using moment matching. We detail the approximation steps and verify its accuracy in Section~\ref{sec:InterferenceApproximation}. Then we derive the optimal power policy that maximizes the uplink capacity in Section~\ref{sec:Optimization} and obtain a closed form expression for the achievable rate. 

\subsection{Capacity Benchmarking}

In order to conduct a reliable and fair comparison that visualizes the performance gain/degradation imposed by FD communication on the uplink capacity, we have to choose a stable HD benchmark. In this paper, we use the HD uplink capacity presented in \cite{ElSawy2014} to benchmark the FD uplink capacity derived.  The HD uplink capacity presented in \cite{ElSawy2014} is given by:
\begin{align}
&\bar{C}_{HD}= \frac{B}{2} \mathbb{E}\left[\log_{2}\left(1+\frac{\rho h}{I_u +N_{0}}\right) \right]
\label{HD_capacity}
\end{align} 
where $\rho$ is a fixed uplink received power at the BS obtained by means of path loss inversion power control. Note that the expectation in \eqref{HD_capacity} is w.r.t. the channel gain $h$ and the uplink to uplink interference $I_u$.  We select the benchmark in \eqref{HD_capacity} for the following four reasons. First, solving an optimization problem in the form of  \eqref{eq:Optimization_problem}, with uplink interference is not tractable due to the coupling between the uplink interference and the optimal power control. Second, the HD performance presented in \cite{ElSawy2014} is based on stochastic geometry analysis with PPP BSs, and hence, it is obtained with a similar system model as the one in this paper. Third, the received power at the BS in \eqref{HD_capacity} is constant and independent of the user location, and hence, can be compared with \eqref{eq:Optimization_problem} which is obtained for a user located at the average distance $\bar{R}$. Fourth, it was shown in \cite[equation (12)]{ElSawy2014} that the capacity in the form of \eqref{HD_capacity} for uplink users is independent of the BSs intensity $\lambda$ and $\rho$, and hence it provide a unified benchmark for the FD operation at different BSs' intensities\footnote{The reason that the uplink capacity is independent from $\rho$ is because the increased received signal power $\rho$ is balanced with an equivalent increase in the uplink interference $I_u$. Similarly, the increased number of interfering UE by increasing $\lambda$ is balanced by an equivalent reduction in the transmission power per UE due to the shorter distance to the serving BS.}. 

As a consequence of the selected benchmark, if the HD capacity, obtained by \eqref{HD_capacity}, exceeds the FD capacity, obtained by \eqref{eq:Optimization_problem}, we can conclude that  FD communication significantly degrades the uplink performance. This is because \eqref{eq:Optimization_problem}  is optimal and ignores uplink interference while \eqref{HD_capacity} in sub-optimal and accounts for the uplink interference. For the same reason,  if the FD capacity, obtained by \eqref{eq:Optimization_problem}, exceeds the HD capacity, by \eqref{HD_capacity}, this does not necessarily indicate a performance improvement provided by FD communication. Therefore, define the FD capacity with fixed transmission power as
\small
\begin{equation}
\bar{C}_{FD}= B \mathbb{E}\left[ \log_2\left(1+ \frac{\overline{P} h}{ \underset{x_i \in \Psi}{\sum} P_{BS} \alpha_i \left\| x_i \right\|^{-\eta}+N_{0} }\right)\right]
\label{FD_fixed}
\end{equation}
\normalsize
The fixed transmit power $\overline{P}$ is the user average transmit power constraint. The capacity defined iin \eqref{FD_fixed} can be used to get a better inference of the performance improvement imposed by FD communication when compared to the HD capacity defined in \eqref{HD_capacity}. It is worth highlighting that the main purpose of this paper is to shed light on the negative impact that may be imposed by FD communication on the uplink performance, and the analysis in this paper perfectly fits this purpose. 

 \section{Downlink to Uplink Interference Approximation}\label{sec:InterferenceApproximation}
 
In this section we obtain approximate distribution of $I$ by matching the gamma distribution parameters to the actual moments of the interference $I$. Although we cannot calculate the {\em pdf} of the aggregate downlink to uplink interference $(I)$, we can still obtain the moments via Laplace transform of the {\em pdf} (LT) of $(I)$, which can be obtained by stochastic geometry analysis for the depicted system model. The {\em pdf} (LT) of $(I)$ is given in the following lemma:

\begin{lemma} \label{lem_LT}
In the depicted system model, the LT of the aggregate downlink to uplink interference in given by
\small
\begin{equation}
\mathcal{L}_{I}(s)=\exp\left(-2\pi\lambda\!\int_{\frac{1}{\sqrt{\pi \lambda}}}^{\infty}\! x\left(  1-\left( 1+\frac{s\Omega  P_{\text{BS}}  x^{-\eta}}{m}\right)^{-m}\right)                \mathrm{d} x    \right).
\label{LT_eq}
\end{equation}
\normalsize
\end{lemma}
\begin{proof}
See \textbf{Appendix~\ref{app_LT}}
\end{proof}

From the LT, the $n^{th}$ moment of the aggregate downlink to uplink interference can be obtained as:
\small
\begin{equation}
\mathbb{E}[I^n]=(-1)^n\frac{\partial^n \mathcal{L}_{I}(s)}{\partial s^n}\Big \vert_{s=0}
\end{equation}
Hence, the first moment is obtained as:
\begin{align}
\mathbb{E}[I]  & = \left. -2\pi\lambda\! \frac{ d \int_{\frac{1}{\sqrt{\pi \lambda}}}^{\infty}\! x\left(  1-\left( 1+\frac{s\Omega  P_{\text{BS}}  x^{-\eta}}{m}\right)^{-m}\right)}{ds } \right|_{s=0}   \notag \\
&\overset{(a)}{=} \left. -2\pi\lambda\!\int_{\frac{1}{\sqrt{\pi \lambda}}}^{\infty}\!\tfrac{\partial \left(x\left(  1-\left( 1+\frac{s\Omega  P_{\text{BS}}  x^{-\eta}}{m}\right)^{-m}\right)\right)}{\partial s}\mathrm{d} x \right|_{s=0} \nonumber\\
& = \frac{2 (\pi\lambda)^\frac{\eta}{2} \Omega P_{\text{BS}}}{\eta-2},
\label{mean}
\end{align}
\normalsize
where  $(a)$ follows from Leibniz integral rule. Similarly, the second moment can be computed as:
\begin{align}
& \mathbb{E}[I^2] =\frac{2 (\pi\lambda)^\eta \Omega^{2} P_{\text{BS}}^{2} }{\eta-2}\left[\frac{2 }{\eta-2}+\frac{ (m+1)(\eta-2)}{2m(\eta-1)}\right].
\label{scnd}
\end{align}
\normalsize

Using the moment matching, the distribution of the downlink to uplink interference is given in the following proposition:
\begin{proposition}  \label{prop}
In the depicted system model, the aggregate downlink to uplink interference can be approximated by a Gamma random variable $Gamma(m_I,\Omega_I)$, where $m_I$ is the shape parameter and $\Omega_I$ is the mean, in which
\begin{equation}
\Omega_{I}=\frac{2 \Omega (\pi \lambda)^{\frac{\eta}{2}} P_{\text{BS}} }{(\eta-2)}, \quad \text{and}
\end{equation}
\begin{equation}
 m_{I} =  \frac{4 (\eta -1) m \Omega }{\eta  (\eta +\eta  m-4)+4-4 (\eta -1) m \Omega }.
\end{equation}
\end{proposition}
\begin{proof}
The proposition is obtained by calculating the gamma distribution parameters via the moments of $I$ are given in \eqref{mean} and \eqref{scnd}, in which $\Omega_I = \mathbb{E}[I]$ and $m_I = \frac{\mathbb{E}[I]^2 }{ \mathbb{E}[I^2]-\mathbb{E}[I]^2 }$.
\end{proof}


The accuracy of the approximation in Proposition~\ref{prop} is verified in Fig.~\ref{fig:MomentMatching_Macro} as well as in Section~\ref{sec:Results}. Looking into Proposition~\ref{prop} we can see that the mean of the aggregate downlink to uplink interference scales with both the BSs intensity and transmit power but with different rates. Since, $\eta > 2$ the effect of $\lambda$ is more prominent than the effect of the transmit power. For instance, for a typical value of $\eta=4$, a FD uplink transmission would have a better performance if the BS transmit power is increased 10 times and the BSs intensity is decreased 4 times.  To this end, since the shape parameter of the distribution of $(I)$ is independent from $\lambda P_{\text{BS}}$ and the mean is directly proportional to $\lambda^\frac{\eta}{2} P_{\text{BS}}$, it can be concluded that the aggregate interference is inversely proportional to the factor $\lambda^\frac{\eta}{2} P_{\text{BS}}$. However, we note that the useful signal power also scales with $\bar{R}^{-\eta} \propto \lambda^\frac{\eta}{2}$.

\begin{figure}[h]
\psfrag{PDF}[c][c][1.2]{~\! PDF}
\psfrag{Monte Carlo}[c][c][0.9]{~\! Monte Carlo}
\psfrag{Gamma}[c][c][0.9]{~\! Gamma}
\begin{center}
\scalebox{0.55}{\includegraphics{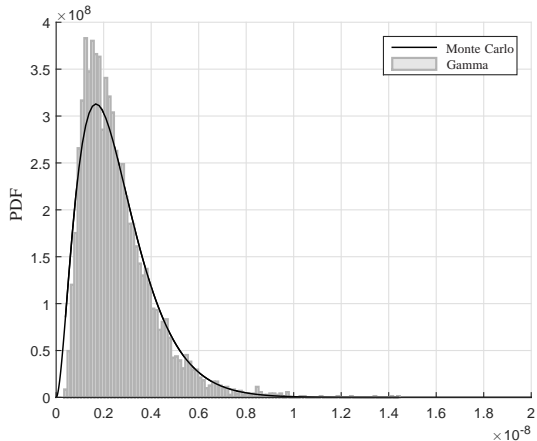}}
\caption{Verification of the approximation of the interference power with Gamma distribution, $r_{0}\!=\!300\text{meters}$, $\lambda\!=\!5 \text{km}^{-2}$, $P_{\text{BS}}\!=\!20\ \text{Watts}$}
\label{fig:MomentMatching_Macro}
\end{center}
\end{figure}

 \section{Optimal capacity under average user transmit power constraint}\label{sec:Optimization}
 
By virtue of the gamma approximation for the {\em pdf} of $I$,  we can derive the optimal uplink power control policy, solve the optimization problem in \eqref{eq:Optimization_problem} and obtain the maximum uplink rate in closed form. It can be observed that the optimization problem in \eqref{eq:Optimization_problem} is  similar to the classic "water-filling" power allocation \cite{Goldsmith05}.  To find the optimal power control policy, we construct the Lagragian of the optimization problem \eqref{eq:Optimization_problem}  as follows
\small
\begin{align}
\mathcal{L}\left(P(h,I),\mu \right)&=\int_{h}\!\int_{I}\left(\mu P(h,I)-B\log_{2}\left(1+\frac{hP(h,I)}{I+N_{0}}\right)\right)\nonumber\\
&\times f_{h}(h)f_{I}(I)\mathrm{d}h\mathrm{d}I\quad-\quad\mu \overline{P},
\end{align}
\normalsize
where $\mu$ is the Lagrange multiplier for the inequality constraint. Since, the optimal transmit power satisfies the Karush-Kuhn-Tucker (KKT) optimal conditions, we have
\begin{equation}
\frac{\partial \mathcal{L}\left(P(h,I),\mu\right)}{\partial P(h,I)}=0.
\end{equation}
Thus,
\begin{equation}
\int_{h}\!\int_{I}\left(\mu - \frac{Bh}{\ln 2\left(I+N_{0}+hP(h,I)\right)}\right)f_{h}(h)f_{I}(I)\mathrm{d}h\mathrm{d}I=0,
\end{equation}
which implies that the optimal power control policy is given by:
\begin{equation}\label{eq:OptimalTransmitPower}
P(h,I)=\left(\frac{B}{\mu_0  \ln2}-\frac{I+N_{0}}{h}\right)^{+}
\end{equation}

The optimal power control policy given in \eqref{eq:OptimalTransmitPower} is a "water-filling" policy, \textit{i.e.} the transmit power and data rate of the user are increased when channel conditions are favorable (less interference and better uplink channel), decreased when channel conditions are not favorable, and the user stops transmitting when 
\begin{equation}
I+N_{0}\geq \frac{Bh}{\mu_0 \ln 2}
\end{equation}
The optimal power control factor $\mu_0$ is obtained solving the constraint at boundary
\begin{equation}\label{eq:LambdaEquation}
\int_{h}\!\int_{I}\left(\frac{B}{\mu_0  \ln2}-\frac{I+N_{0}}{h}\right)^{+} f_{h}(h)f_{I}(I)\mathrm{d}h\mathrm{d}I = \overline{P}
\end{equation}

Combining \eqref{eq:averageCapacity} and \eqref{eq:OptimalTransmitPower}, the maximum achievable rate can be expressed as:
\begin{equation}\label{eq:optimalCapacity}
C^{*}_{FD}=\frac{B}{\ln 2}\int_{h}\!\int_{I}\ln\left(\frac{hB}{(I+N_{0})\mu_0 \ln 2}\right) f_{h}(h)f_{I}(I)\mathrm{d}h\mathrm{d}I,
\end{equation}
where $\frac{hB}{(I+N_{0})\mu_0 \ln 2} \geq 1$.  To simplify the calculation of the maximum achievable rate, we characterize the resultant distribution of the channel-to-interference-plus-noise-ratio (CINR) in the following lemma:

\begin{lemma} \label{lem_gam}
 Let $\gamma = \frac{h}{I+N_o}$ be the resultant CINR for the depicted system model, then $x$ has the following distribution
 
 \begin{equation} \label{eq:PDF_Ration_Gamma}
 f_{\gamma}(x)=\frac{k^{m_{0}}}{\mathcal{B}\left(m_{0},m_{I}\right)}\left(1+kx\right)^{-m_{0}-m_{I}}x^{m_{0}-1},
\end{equation}
where $\mathcal{B}(a,b) = \int_0^1 t^{a-1} (1-t)^{b-1} dt$ is the beta function, and $k=\frac{(2\sqrt{\lambda})^{\eta}m_{0}(\Omega_{I}+N_{0})}{m_{I}\Omega_{0}}$.

 \end{lemma}

This is obtained using the method of bivariate random variable transformation considering the summation and the ration of two Gamma random variables with different mean and scale parameter.

Exploiting Lemma~\ref{lem_gam},  we can rewrite the maximum achievable rate in \eqref{eq:optimalCapacity} as 
\begin{equation}\label{eq:optimalCapacity2}
C^{*}_{FD}=\frac{B}{\ln 2}\int_{\frac{1}{a_{0}}}^{\infty}\!\ln\left(a_{0}x\right) f_{x}(x)\mathrm{d}x,
\end{equation}
where $a_{0}=\frac{B}{\mu_0 \ln 2}$. The maximum achievable rate is then characterized via the following theorem:

\begin{theorem}
In the depicted system model, the maximum achievable uplink rate under full-duplex communication is upper-bounded by:
\small
\begin{equation}\label{optimalCapacity3}
C^{*}_{FD}=\frac{Ba_{0}^{m_{I}}{}_3\text{F}_2\left(m_{I},m_{I},m_{0}+m_{I}; 1+m_{I},1+m_{I};-\frac{a_{0}}{k}\right)}{\mathcal{B}\left(m_{0},m_{I}\right)m_{I}^{2}k^{m_{I}}\ln 2}.\end{equation}
\normalsize
where $a_{0}$ is numerically calculated by solving the following equation
\small
\begin{align}\label{eq:Expression_Alpha0}
\frac{a_{0}^{m_{I}+1}}{\mathcal{B}\left(m_{0},m_{I}\right)k^{m_{I}}}\Bigg(\frac{{}_2\text{F}_1 \left(m_{I},m_{I}+m_{0},1+m_{I},-\frac{a_{0}}{k}\right)}{m_{I}}\nonumber\\
 -\frac{{}_2\text{F}_1 \left(m_{I}+1,m_{I}+m_{0},2+m_{I},-\frac{a_{0}}{k}\right)}{m_{I}}\Bigg)=\overline{P}
\end{align}
\normalsize

\end{theorem}
\begin{proof}
See \textbf{Appendix~\ref{app:integral}}
\end{proof}
Theorem 1 gives an upper bound on the capacity of uplink communication under FD operation, which is main contribution of the paper. This upper-bound can be used to quantify  the potential gains to be expected from FD communication, under a certain network setup, in the uplink direction. It is important to highlight that the obtained upper bound is obtained by assuming that the user has perfect instantaneous information about the aggregate downlink to uplink interference, which may not be practical. Nevertheless, this upper bound can benchmark the performance of more practical yet suboptimal power control techniques as shown in the next section. 



 \section{Simulation and Numerical  Results}\label{sec:Results}
 
 In this section, we evaluate the FD performance in two network setups, namely, the macro-cells and the micro-cells network. The macro-cells setup can represent the cellular network in sub-urban areas in which the BSs have high transmit power and are sparse in the  spatial domain.  In contrast to the macro-cells, the micro-cells can represent the cellular network in urban areas in which the BSs have lower transmit powers and are denser in the spatial domain. In each case,  we plot the upper-bound obtained in Theorem~1 against the conventional HD performance obtained in \cite{ElSawy2014} to visualize the maximum performance gain that can be harvested from FD communication. We also conduct Monte Carlo simulations to verify the upper-bound obtained in Theorem~1. 
Unless otherwise stated, the network parameters are selected as follows; the noise and uplink interference power is set to $N_{0}=10^{-9}$, the path loss exponent $\eta=4$. The total bandwidth is $B=180$KHz, and the parameters of the fading for the interfering channel is  $(m_{I},\Omega_{I})=(1,1)$, and for the reference user channel $(m_{h},\Omega_{h})=(2,(2\sqrt{\lambda})^{\eta})$.

In the first simulation example, given in Fig.~\ref{fig:Capacity_power_micro}, we study the variation of the uplink capacity with the transmit power of the BSs in a micro-cells environment.   In this case study, the FD can provide a tangible performance improvement compared to HD, the gain in uplink capacity is up to $400\%$ higher than the HD uplink for the optimized FD, and $350\%$  for the FD with fixed transmit power. However, the figure shows that  as expected, the FD performance gain monotonically decays with the transmission power of the BSs due to the increased uplink interference. It is worth noting that the HD capacity is fixed with the BS transmission power as the HD operation does not impose cross mode interference as indicated in \eqref{HD_capacity}.  The figure also shows that the gain w.r.t. the optimal FD compared to the fixed power FD enhances with the increase of the BSs transmit power,  which indicates the importance of estimating the downlink interference in order to achieve higher performance.


The second simulation example, given in Fig.~\ref{fig:Capacity_power_macro}, shows the uplink capacity in a macro-cells environment which is 10 times less dense but can have up to 40 times higher transmission power  compared to the micro BSs. In this case, the figure shows that the FD communication  is beneficial only when both BSs transmit power and density are low, otherwise the uplink capacity is deeply affected.

 
\begin{figure}[h]
\psfrag{FD-WF-IA}[c][c][0.9]{~\!$C^{*}_{FD}$}
\psfrag{FD-WF-IA-MC}[c][c][0.9]{~\!$C^{*}_{FD}$-MC}
\psfrag{FD-[7]}[c][c][0.9]{~\!$\quad\ \bar{C}_{FD}$}
\psfrag{HD-[7]}[c][c][0.9]{~\!$\quad\ \bar{C}_{HD}$}
\psfrag{lambda1}[c][c][0.9]{~\!$\lambda=10^{-5}$}
\psfrag{lambda2}[c][c][0.9]{~\!$\lambda=5*10^{-5}$}
\psfrag{User uplink capacity (kbps)}[c][c][1.1]{~\! Uplink Capacity (kbps)}
\psfrag{BS-Power}[c][c][1.2]{~\! $P_{\text{BS}}$(Watts)}

\begin{center}
\scalebox{0.55}{\includegraphics{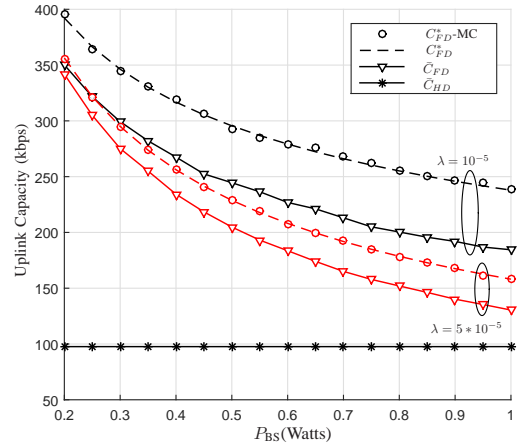}}
\caption{Variation of the uplink capacity according to the transmit power at the BSs  for the case of microcells, $\lambda=5\ 10^{-5}m^{-2}$}
\label{fig:Capacity_power_micro}
\end{center}
\end{figure}
\begin{figure}[h]
\psfrag{FD-WF-IA}[c][c][0.9]{~\!$C^{*}_{FD}$}
\psfrag{FD-WF-IA-MC}[c][c][0.9]{~\!$C^{*}_{FD}$-MC}
\psfrag{FD-[7]}[c][c][0.9]{~\!$\quad\ \bar{C}_{FD}$}
\psfrag{HD-[7]}[c][c][0.9]{~\!$\quad\ \bar{C}_{HD}$}
\psfrag{lambda1}[c][c][0.9]{~\!$\lambda=10^{-6}$}
\psfrag{lambda2}[c][c][0.9]{~\!$\lambda=5*10^{-6}$}
\psfrag{User uplink capacity (kbps)}[c][c][1.1]{~\! Uplink Capacity (kbps)}
\psfrag{BS-Power}[c][c][1.2]{~\! $P_{\text{BS}}$(Watts)}
\begin{center}
\scalebox{0.55}{\includegraphics{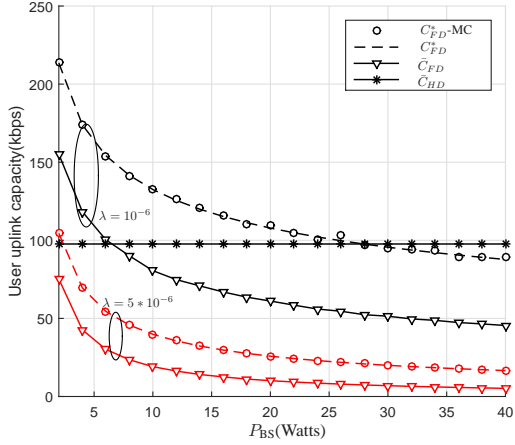}}
\caption{Variation of the uplink capacity according to the transmit power at the BSs  for the case of macrocells, $\lambda=5\ 10^{-6}m^{-2}$}
\label{fig:Capacity_power_macro}
\end{center}
\end{figure}

\section{Conclusion} \label{Conclusion}
In this paper, the mean and shape parameter of the downlink to uplink interference power are derived  as function of the network parameters and it has been discovered that the density of the BSs has stronger  effect on the  downlink to uplink interference than the transmit power of the BSs. The results show that for micro-cells, in spite of the high vulnerability of the uplink against interference, FD communication can provide significant gain in terms of uplink data rate compared to HD communication. However, for macro-cells, FD communication is essentially harmful to the uplink transmission, unless both BSs density and transmit power are very low.
\vskip -30pt
\appendices
\renewcommand{\thesubsection}{\Alph{subsection}}
\section{Proof of Lemma~\ref{lem_LT}} \label{app_LT}
\vskip -10pt
\begin{align}
\mathcal{L}_{I}(s)&=\mathbb{E}\left[e^{-Is}\right]\nonumber=\mathbb{E}\left[\exp\left(-s\sum_{i\in \psi}P_{\text{BS},i}\alpha_{i}\vert x_{i}\vert^{-\eta}\right)\right],\nonumber\\
&=\mathbb{E}\left[\prod_{i\in \psi}\exp\left(-sP_{\text{BS},i}\alpha_{i}\vert x_{i}\vert^{-\eta}\right)\right].
\end{align}
Using probability generating functional (PGFL) \cite{Haenggi2009}, we get
\begin{equation}
\mathcal{L}_{I}(s)=\mathcal{G}\!\left(\mathbb{E}_{\alpha}\!\left[e^{-sP_{\text{BS}}\alpha\vert x\vert^{-\eta}}\right]\right).
\end{equation}
Assuming that the base station placement follows a Poisson Point Process (PPP) with density $\lambda$, with a minimum interfering distance $\frac{1}{\sqrt{\pi\lambda}}$~\cite{Heath2013}, we get
\small 
\begin{align}
\mathcal{L}_{I}(s)&=\exp\left(-2\pi\lambda\!\int_{\frac{1}{\sqrt{\pi\lambda}}}^{\infty}\! x\left(  1-\mathbb{E}_{\alpha}\!\left[e^{-sP_{\text{BS}}\alpha x^{-\eta}}\right]\right)                \mathrm{d} x    \right)\nonumber\\
&=\exp\left(-2\pi\lambda\!\int_{\frac{1}{\sqrt{\pi\lambda}}}^{\infty}\! x\left(  1-\mathcal{M}_{\alpha}\!\left(sP_{\text{BS}} x^{-\eta}\right)\right)                \mathrm{d} x    \right).
\end{align}
\normalsize
$\mathcal{M}_{\alpha}(.)$ is the Moment Generating Function of the random variable $\alpha$. Considering that the fading coefficient $\alpha$ follows a Gamma distribution $\alpha \sim Gamma(m, \Omega)$, we get~\eqref{LT_eq}.

\section{Proof of Theorem 1}\label{app:integral}
Expanding \eqref{eq:optimalCapacity2}, the expression of the maximum achievable rate becomes
\footnotesize                
\begin{equation}\label{optimalCapacity3}
C^{*}=\frac{Bk^{m_{h}}}{\mathcal{B}\left(m_{h},m_{I}\right)\ln 2}\int_{\frac{1}{a_{0}}}^{\infty}\!\ln\left(a_{0}x\right)\left(1+kx\right)^{-m_{h}-m_{I}}x^{m_{h}-1} \mathrm{d}x,
\end{equation}
\normalsize
This integral can be solved with integration by part considering $u=\ln(a_{0}x)$ and $\mathrm{d}v=\left(1+kx\right)^{-m_{0}-m_{I}}x^{m_{0}-1}$. Therefore, $\mathrm{d}u=\frac{1}{x}$ . Using  \cite[Eq.~(3.194)]{Gradshteyn2007} to find $v$, then the relation given in \cite[Eq.~(15.3.7)]{Abramowitz1964},  and  finally  applying  integration by parts, we obtain
\small       
\begin{align}
I&=\left[\frac{\Gamma(1+m_{0})\Gamma(m_{0})}{{m_{0}k^{m_{0}}}\Gamma(m_{0}+m_{I})}\left(\ln(a_{0}x)-\ln(x)\right)\right]_{x=\frac{1}{a_{0}}}^{\infty}\nonumber\\
&+\frac{a_{0}^{m_{I}}{}_3\text{F}_2\left(m_{I},m_{I},m_{0}+m_{I}; 1+m_{I},1+m_{I};-\frac{a_{0}}{k}\right)}{m_{I}^{2}k^{m_{0}+m_{I}}}\nonumber\\
&=\frac{a_{0}^{m_{I}}{}_3\text{F}_2\left(m_{I},m_{I},m_{0}+m_{I}; 1+m_{I},1+m_{I};-\frac{a_{0}}{k}\right)}{m_{I}^{2}k^{m_{0}+m_{I}}}
\end{align}
\normalsize

\bibliographystyle{IEEEtran}
\nocite{*}
\bibliography{Fundamental_Limit_FullDuplex_HetNet}


\end{document}